\documentclass{article}
\usepackage[utf8]{inputenc}
\usepackage{graphicx, amsmath, amsthm, amssymb}
\usepackage[lined,ruled,commentsnumbered]{algorithm2e}
\usepackage{blindtext, rotating}
\usepackage{enumerate}

\title{The Emergence of League and Sub-League Structure in the Population Lotto Game}
\author{Giovanni Artiglio, Aiden Youkhana, Joel Nishimura}
\date{The School of Mathematical and Natural Sciences \\ Arizona State University}

\newtheoremstyle{case}{}{}{}{}{}{:}{ }{}
\newtheorem{theorem}{Theorem}[section]
\newtheorem{thm}{Theorem}[theorem]
\newtheorem{lemma}{Lemma}[theorem]

\newtheorem{definition}{Definition}[theorem]
\newtheorem{corollary}{Corollary}[theorem]
\newtheorem{conjecture}{Conjecture}[theorem]

\usepackage[usenames,dvipsnames]{xcolor}
\begin{document}

\maketitle

\section{Abstract}
In order to understand if and how strategic resource allocation can constrain the structure of pair-wise competition outcomes in competitive human competitions we introduce a new multiplayer resource allocation game, the Population Lotto Game. This new game allows agents to allocate their resources across a continuum of possible specializations. While this game allows non-transitive cycles between players, we show that the Nash equilibrium of the game also forms a hierarchical structure between discrete `leagues' based on their different resource budgets, with potential sub-league structure and/or non-transitive cycles inside individual leagues. We provide an algorithm that can find a particular Nash equilibrium for any finite set of discrete sub-population sizes and budgets. Further, our algorithm finds the unique Nash equilibrium that remains stable for the subset of players with budgets below any threshold.

\section{Introduction}

Many human societies rely upon the outcomes of various human-human competitions to organize their society and allocate resources. 
The clearest and somewhat idealized forms of these competitions are sports, whose win-loss structure offers an interesting theoretical test bed. However, while there are many possible patterns for wins-losses across a population, organization, association or league, it appears that many sport outcomes are well explained by a simple hierarchical ranking. In this paper we develop a model for how a population of competitors allocate fixed resources, in order to explain the emergence of a strict hierarchy or the segmentation of competition into discrete competitive categories in competition win-loss outcomes.

A common hobby in the United States (particularly inside the statistics departments of some colleges and universities) is the construction of `March Madness' brackets, wherein people attempt to predict the outcome of the NCAA basketball tournament. Experience has proven that correctly predicting the exact outcomes is extremely difficult. While the difficulty of exactly predicting the outcome is doubtlessly a consequence of the large number of potential outcomes, it is often due to the number of `upsets' where teams regarded as worse beat those regarded as better. Indeed, it only takes a few upsets to make the estimation task difficult. A natural question regarding these upsets is whether they actually reveal hidden structure about what teams are competitive against another, reveal information about how a team has changed since the regular season, or are simply stochastic artifacts reflecting the noisy nature of basketball competitions. 

Classical methods for inferring competitor strength from win-loss data have typically assumed that strength is one-dimensional and transitive. For instance, the Thurstone model \cite{thurstone1927law}, ELO-scores developed for chess, and Bradly-Luce-Terry \cite{bradley1952rank} inference methods each describe a competitor's competitiveness with a single number and have generally proven able to produce reasonable estimates. More recent models, taking advantage of the increased size and availability of competition data and increased computational resources have begun to investigate if there are multi-dimensional characterizations that could improve on these classical results by capturing non-transitive competition cycles, such as the cycle present in the popular game `rock-paper-scissors' (where rock beats scissors, which in turn beats papers, and where paper beats rock). Indeed, the presence of such cycles has often been speculated by some fans, who imagined that some features of teams were particularly well suited for some match-ups. For the most part, it does not appear that human competition data of individual competitors or fixed teams contains robust and predictive cycles \cite{chen2016modeling,makhijani2019parametric,chen2016predicting}, with a limited number of possible exceptions: the computer game Starcraft \cite{chen2016modeling}, Street Fighter IV \cite{makhijani2019parametric}, some voting scenarios \cite{makhijani2019parametric} or when match-ups are subject to different contexts (i.e. weather, home/away, etc.) \cite{chen2016predicting}. In contrast, competitions involving spontaneously created teams, such as in the online match making of League of Legends \cite{claypool2015surrender,chen2017player} and Ghost Recon \cite{delalleau2012beyond} have match results that may depend in a complex way on the skills and attributes of the individual team members, leading to cycles. Nonetheless, given that almost all games could conceivably allow for non-transitive cycles, it is surprising that they are seemingly rare in the data. 

One hypothesis for the absence of the robust non-transitive cycles in professional and semi-professional human competitions is that at high levels of play competitors intentionally train by focusing on any weaknesses in their play, and by doing this, they ensure that the resulting competitive structure appears hierarchical. In order to investigate this, we propose a new game, the `Population Lotto Game', where an infinite population of players must strategically allocate their budgeted resources across a continuous space of possible specializations. 

Our main result is that the Nash equilibrium of the Population Lotto Game spontaneously creates a hierarchical relationship between different budgets. In particular the Nash equilibrium separates players into disjoint `leagues' wherein competitors in higher leagues always defeat those lower leagues. Inside of each of these leagues there is the possibility of both non-transitive cycles and/or further hierarchical relationships between `sub-leagues'. In each of these cases, we prove conditions that the average overall population investment of a Nash equilibrium must obey. We further show how to produce a Nash equilibrium for any finite discrete set of sub-populations and associated budgets.  

The organization of the paper is as follows. In section 3, 'Related Work,' we recapitulate established results of the Colonel Lotto game and the Continuous General Lotto game. In section 4, we introduce our 'Population Lotto Game,' and demonstrate that it can support non-transitive outcomes between different portions of the populations. In section 5, we prove several theorems that characterize the Nash Equilibria of the Population Lotto Game. In section 6, we use our characterization of the Nash Equilibria to show how to explicitly construct Nash equilibria for discretized budget distributions. In section 7 we introduce terminology to describe the internal population structure of these Nash equilibria, revealing the ways in which these equilibria allow, or disallow non-transitive interactions. Finally, in section 8 we conclude with a brief discussion of potential extensions.

\section{Related Work}

Game theoretic models of resource allocation have a rich history. The original Colonel Blotto game was introduced by Emile Borel in the 1920's \cite{borel1953theory} and in it two colonels allocate their soldiers across a $K$ different battlefields, the goal being to achieve numerical superiority at as many battlefields as possible.
Despite, or perhaps because of its martial framing it has since been used and modified to model resource allocation problems in democratic elections \cite{sahuguet2006campaign,roberson2006colonel}, even including the 2016 US presidential elections \cite{lee2018using}. 

There are numerous generalizations of the original Colonel Blotto game, including when the budget constraints are relaxed, whether battlefields are discrete or continuous, whether battlefields have similar payoffs along with many others \cite{kovenock2020generalizations,hart2008discrete}. When the battlefields are indistinguishable, the symmetry of battlefields allows for simplified mixed Nash strategies where players first divide their resources into $K$ partitions, and then assign each partition to a random battlefield. In such a setting, since neither player knows which of their partition will face which of the opponent's, the game becomes the same as the `dice design game', where players allocate a finite and fixed number of pips to the faces of a die after which each player then rolls their die and the player with the highest resulting number wins, as studied in \cite{de2006optimal,de2009optimal}. Symmetric battlefields also allow for a continuum of battlefields, as in the Continuous General Lotto game, where two players each pick distributions over positive real numbers so that the mean of each distribution is constrained to match each player's fixed budget. The winner is determined by drawing a random number from each player's distribution, with the larger number winning \cite{hart2008discrete}.

Recently there has been interest in multiplayer versions of these games. For instance, two player solutions can be used to solve a multiplayer Colonel Blotto game when players are symmetric \cite{kashyap2020essay}. Alternatively, a model for multiplayer General Lotto was analyzed, wherein players simultaneously play against all others \cite{boixadsera2021multblotto}. In contrast, we consider a modification to the Continuous General Lotto game where players compete in randomly assigned pair-wise competitions. For our pair-wise competitions, we build off of the Continuous General Lotto game where players $A$ and $B$ each choose distributions for their non-negative random variables $X_A$ and $X_B$ constrained so that $\mathbb{E}[X_A]=a$ and $\mathbb{E}[X_B]=b$ where $a$ and $b$ are their respective budgets and $a \geq b$ \cite{bell2011competitive,hart2008discrete}. In the Continuous General Lotto game, it has been shown \cite{sahuguet2006campaign} that the unique Nash equilibrium has probability density functions $f_A$ and $f_B$ for $X_A$ and $X_b$ respectively as:
\begin{eqnarray*} 
f_A &=& U(0,2a),  \\ 
f_B &=& (1-b/a) \delta_x(0) + (b/a)U(0,2a),  
\end{eqnarray*}
where $U(c,d)$ denotes the uniform distribution on the interval $(c,d)$ and $\delta_x(0)$ denotes the Dirac measure.
The strategies can be interpreted as: player A simply plays a uniform distribution with expectation $a$ on the interval $(0,2a)$; player B will try to play the same distribution as A, but because their budget is lower they can only afford this by ``sacrificing" some games by playing 0 with probability $1-b/a$.

Interestingly, the solution to the Continuous General Lotto game informs the class of solutions to the dice design game, in that the Nash equilibria of the dice design game tend to resemble discrete approximations of the uniform distribution solution of the Continuous General Lotto game. For instance, for a six sided die with a pip budget of 21 the `standard die' with pip allocation $[1,2,3,4,5,6]$ is a Nash equilibrium and resembles a discrete approximation of a uniform distribution with mean $3.5$ (conditioned on not allowing $0$ pip faces).    

In these games, the Nash equilibrium solutions are uniform distributions or discrete approximations of uniform solutions because these distributions have no particular region of higher or lower density to take advantage of. Indeed, strategies that are non uniform because they include a spike or region of above average density have a clear weak-point that can be easily exploited. For instance a die that has far too many 1 and 4 pip faces would lose to a die that specialized in 2 and 5 faces, and a distribution which includes a delta spike at $x$ loses to facsimile with a delta spike at $x+\epsilon$, paying for that extra investment with a much smaller delta spike at $0$. Notice that this implies that best response dynamics (or approximate best response when the optimal response doesn't exist) in these games tends away from Nash equilibrium play, and is fundamentally related to these game's ability to capture a sequence of mutually non-transitive strategies. 

Indeed, that these games allow for rich sets of nontransitive strategies is why they serve as an excellent means to explore the effect of strategic budgeting on the transitivity of match-up outcomes in a population.

\section{The Population Lotto Game}
Consider the Continuous General Lotto Game as introduced before, but expand the number of players to $n \in \mathbb{N}$ as follows. Player $i$, $1 \leq i \leq n$, chooses a probability density function $f_i:\mathbb{R}^+ \to \mathbb{R}^+$ constrained according to their budget $b_i>0$ so that the mean of $f_i$ is $b_i$. Each player is then randomly paired against another player and the payoff between matched players proceeds as in the Continuous General Lotto Game. Namely, if player $i$ is paired with player $j\neq i$, and $x_i\sim f_i$ and $x_j\sim f_j$ then $i$ beats $j$ if $x_i>x_j$, and $j$ beats $i$ if $x_j>x_i$ with ties broken uniformly at random. Let $H(f_i,f_j)$ denote the payoff to player $i$ from playing player $j$. Thus, as in the continuous Lotto game, so long as $f_i$ and $f_j$ have no overlapping atoms $H(f_i,f_j)=\int_0^\infty f_i(x) \int_0^x f_j(y)dy dx  $. Denote the expected payoff against all other players as, $H(f_i) = \mathbb{E}[H(f_i,f_j)] = \frac{1}{n}\sum_j H(f_i,f_j) $.
If $n=2$, then this is exactly the Continuous General Lotto Game as discussed before. 
For $n>2$, we call this new game the pair-wise\footnote{Another generalization of the blotto Game \cite{boixadsera2021multblotto} assumes all-to-all rather than Pair-Wise multiplayer competitions} Multiplayer Lotto Game (hereafter, simply the Multiplayer Lotto Game). Notice that in the Multiplayer Lotto Game, players know the budgets of all potential competitors, but each player must allocate their budget before knowing exactly who their competitor will be. We consider the scenario where competitors choose their strategies in an attempt to maximize their expectation of winning over all possible match-ups.

In order to investigate this game for large populations we introduce the `Population Lotto' game as follows. Consider the $n \rightarrow \infty$ limit of the Multiplayer Lotto game, where there is a continuum of players $i$ with budgets $b_i$ distributed according to a budget probability density function $\mathcal{B}(y)$. We interpret the distribution $\mathcal{B}(y)$ as describing the underlying distribution of resources in the population, be that money, time or innate skill. For the collection of players with budget $y$, denote $\tilde{f}_y=E(f_i|b_i=y)$ as the average strategy of all players with budget $y$. Based on these averages, denote the total population cumulative density function as, 
\begin{equation*}
    G(x)= \int_0^\infty \int_0^x  \mathcal{B}(y)\tilde{f}_y(z) dzdy.
\end{equation*} 
Since the population is infinite, each individual $f_i$ has negligible influence on $G(x)$, simplifying the expected payoff function for $i$,   \begin{eqnarray*}
H(f_i) &=& \mathbb{E}[H(f_i,f_j)] \\
&=& \int_0^\infty \mathcal{B}(y) \int_0^\infty f_i(x) \int_0^x \tilde{f}_y(z)dz dx dy \\
&=& \int_0^\infty f_i(x)G(x) dx.
\end{eqnarray*}
Thus the average population cumulative density function $G$, and subsequent strategy $g(x)=G'(x)$, is the fundamental consideration for the competitive decisions of each player in the Population Lotto game, where playing against the population is analogous to playing against a single player whose distribution is $g$, in that $H(f_i) = H(f_i,g)$. That each player in the Population Lotto game opposes the same aggregate strategy $g$ is the primary simplifying feature of the Population Lotto game in comparison with the Multiplayer Lotto game, where each player $i$ would analogously see $g-\frac{1}{n}f_i$ as their own individually varying competitor. 

One of the most important properties of the Population Lotto game is its ability to support non-transitive match-ups. In fact, the Population Lotto game can represent any non-transitive set of dice from the dice design game as a set of player distributions. For instance, consider in Figure \ref{fig:example_atom} three dice $A$, $B$, and $C$ with $30$ pip allocations, wherein the non-transitive cycle is that $B$ beats $A$, $C$ beats $B$ and $A$ beats $C$ each with probability $\frac{5}{9}$. This structure can be captured as as piece-wise uniform probability distributions as in Figure \ref{fig:example_atom}, where the distributions corresponding to dice $A$, $B$ and $C$ all have equal budgets and maintain the same win-loss probabilities. In this particular case the distributions are a Nash equilibrium in the Population Lotto game.       

\begin{figure}[hbt!]
    \centering
    \begin{turn}{45} 
    \includegraphics[width=.2\textwidth]{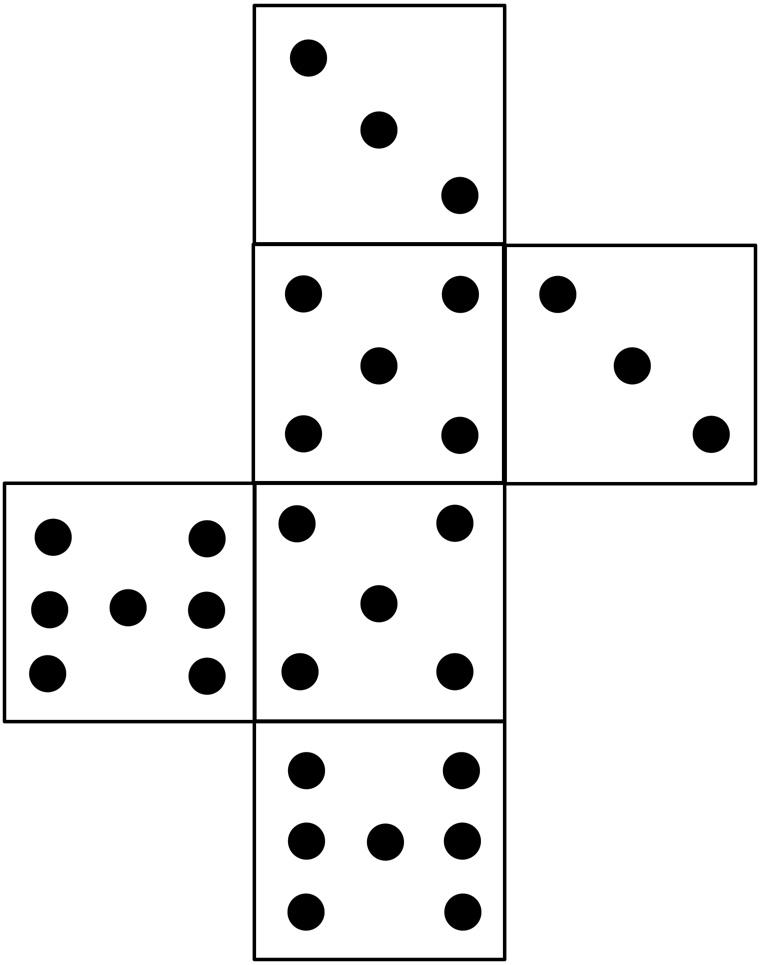}
    \end{turn}
    \includegraphics[width=.49\textwidth]{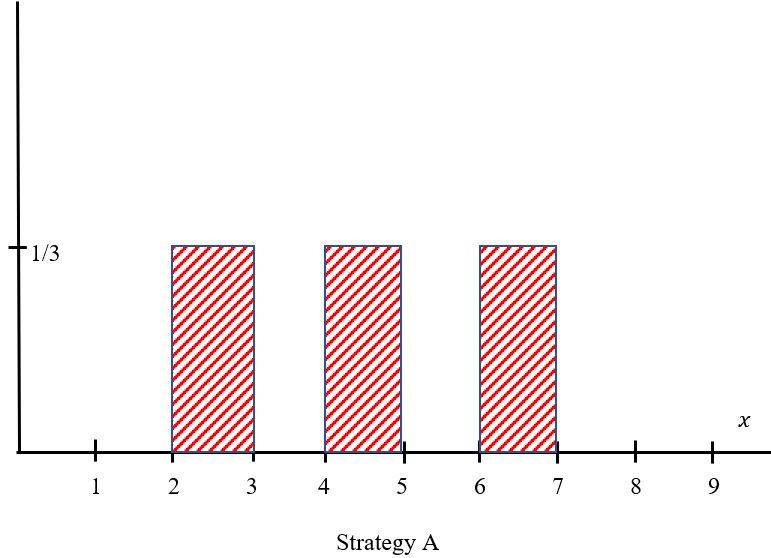} \\
    \begin{turn}{45} 
    \includegraphics[width=.2\textwidth]{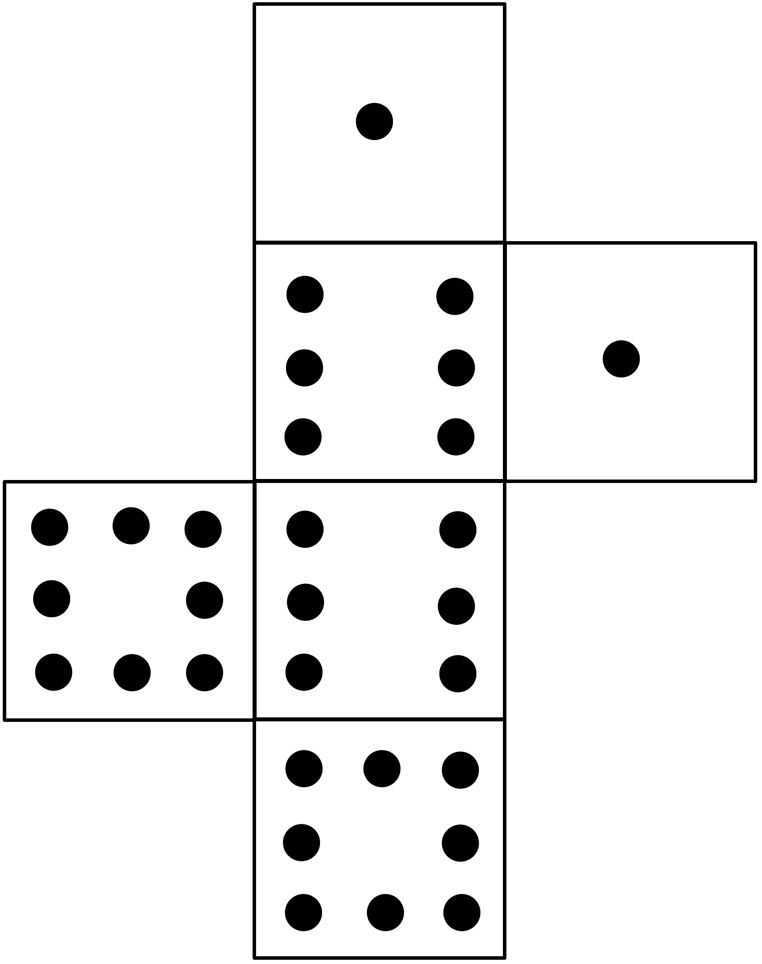}
    \end{turn} 
    \includegraphics[width=.49\textwidth]{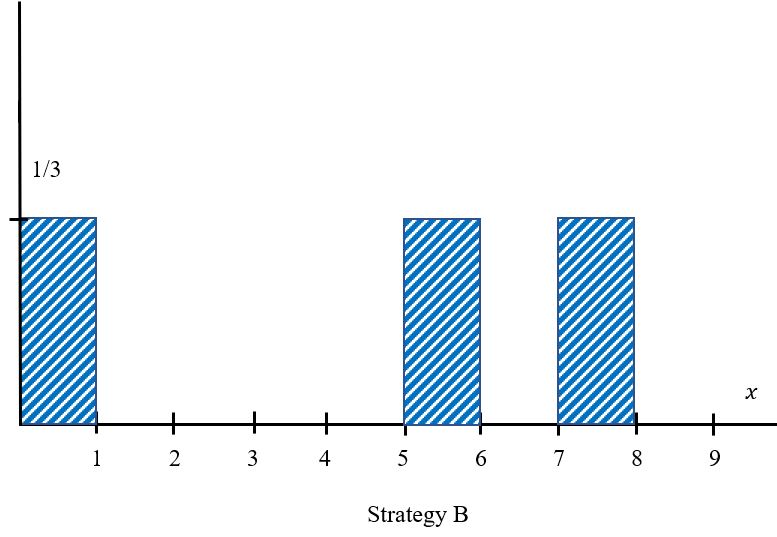} \\
    \begin{turn}{45} 
    \includegraphics[width=.2\textwidth]{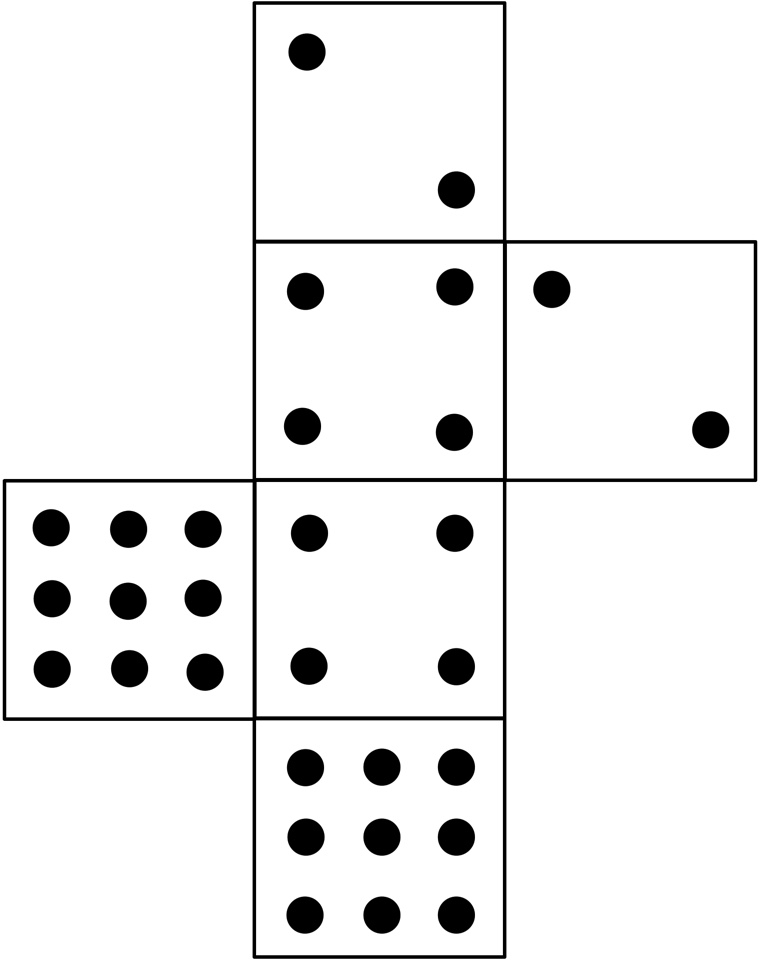}
    \end{turn}
    \includegraphics[width=.49\textwidth]{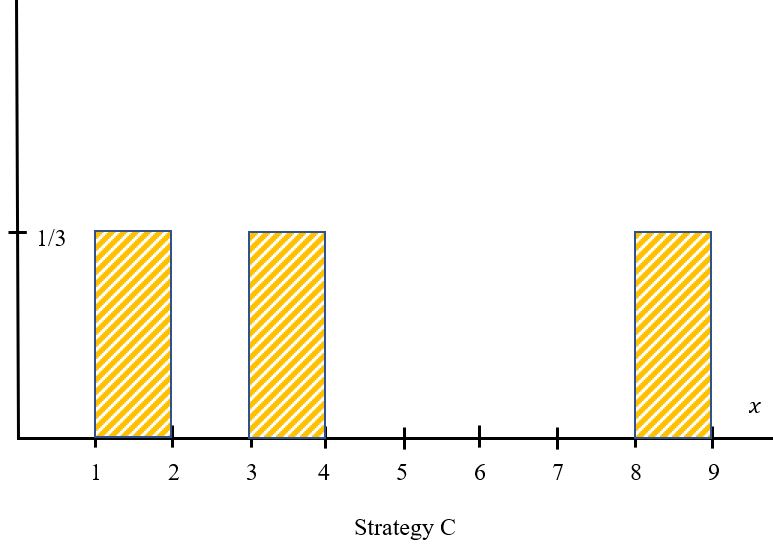}

    \caption{A set of non-transitive dice, each with a total of 30 pips, where $B$ beats $A$, $C$ beats $B$ and $A$ beats $C$ can be mapped to distributions of the Population Lotto Game with the same pairwise outcomes. A population equally divided into A, B and C is a Nash equilibrium, because all combined they would create a uniform distribution.}
    \label{fig:example_atom}
\end{figure}

\section{Nash Equilibria}

In order to simplify our analysis, we will first consider the simplest nontrivial strategies that satisfy budget constraints, `dyads,' a pair of dirac deltas scaled to satisfy the budget constraint. For budget $b$, let dyad $\chi(x|\theta) $ with parameters $\theta = \{x_1,x_2,b \}$, be composed of a pair of Dirac delta functions $\delta$ that satisfy the budget constraints such that:
\begin{eqnarray*}
\chi(x|\theta)  &=& \lambda \delta(x-x_{1}) + (1-\lambda)\delta(x-x_{2}) \\
b & = & \lambda x_{1} +(1-\lambda) x_{2} ,
\end{eqnarray*}
where $0\leq x_{1}<b<x_{2}$ and $\lambda = \frac{x_{2}-b}{x_{2} - x_{1}}$. For a given population distribution $g$ the payoff of dyad $\chi$ is:
\begin{eqnarray*}
H(\chi) &=& \lambda G(x_1) + (1-\lambda)G(x_2) \\
 &=& \frac{b-x_1}{x_2 - x_1}G(x_2)+\frac{x_2-b}{x_2 - x_1}G(x_1)
\end{eqnarray*}
Because dyads satisfy budget constraints they serve as the perfect perturbation to check whether a distribution is optimal. Let $\sigma(f)$ denote the support of $f$. If for dyad $\chi$, 
$H(\chi)>H(f_i)$ then $f_i$ is not optimal. Notice, that this statement remains true whether or not distributions are allowed to contain atoms, or even if strategies are required to be continuous (which we do not assume), since a dyad $\chi(x| \{x_1,x_2,b\})$ can be approximated by a continuous function supported on $[x_{1},x_{1}+\epsilon]\cup[x_{2},x_{2}+\epsilon]$ that satisfies the budget constraint. Furthermore, any number in $[x_{1},x_{1}+\epsilon]$ or $[x_{2},x_{2}+\epsilon]$ will win in any situation where $x_{1}$ or $x_{2}$ would win respectively, ensuring that the payoff of this approximation matches the payoff of $\chi$ as $\epsilon\to0$.  

Moreover, in order for a distribution $f_i$ to optimize $H$ there must be an $\tilde{f}_i$ equal almost everywhere such that for any dyad $\chi (x| \{x_1,x_2,b_i\})$, if $\sigma(\chi)\subseteq \sigma(\tilde{f}_i)$ then $\chi$ is optimal. To see this: suppose not, that $\tilde{f}_i$ has support and positive measure on an interval $[x_1-\delta,x_1]$ and $[x_2-\delta, x_2]$ for $x_1<b_i<x_2$ for all $\delta>0$, and that a dyad $\chi_i$ supported on $x_1$ and $x_2$ is not an optimal dyad. Since $\chi_i$ is not an optimal dyad for budget $b_i$ there exists $\hat{\chi_i}$ such that $H(\hat{\chi_i})= H(\chi_i)+\epsilon$ for some $\epsilon>0$, but this immediately suggests a perturbation to $\tilde{f_i}$ that improves it. Namely let $\hat{f}_i$ be a distribution that reassigns a positive portion of the measure assigned to $[x_1-\delta,x_1]$ and $[x_2-\delta, x_2]$ to $\hat{\chi_i}$, ensuring that $\hat{f}_i$ satisfies the budget $b_i$ but has a higher expected outcome. Thus, a distribution $f_i$ is optimal if and only if there is an almost everywhere equivalent distribution $\tilde{f}_i$, such that every possible dyad in the support of $\tilde{f}_i$ is optimal. 

Thus, whether a distribution optimizes $H$ can be determined solely by considering dyads. As the next lemma demonstrates, assuming that a dyad with support on $x_1$ and $x_2$ is optimal implies a constraint on the remaining shape of $G(x)$. In particular, the benefit per resource for playing $x$, $\frac{G(x)}{x}$, must be bounded by the average such benefit $\frac{H(\chi_i)}{b_i}$ plus an interpolation between between the per resource gain at $x_1$ and $x_2$:

\begin{lemma}
\label{lem:dyad_bound}
Player strategy dyad: $\chi_i (x| \{x_1,x_2,b_i\})= \lambda_i \delta(x_1) + (1-\lambda_i)\delta(x_2)$ is optimal if and only if for all $x$, $\frac{G(x)}{x} \leq \frac{H(\chi_i)}{b_i} + (\frac{1}{x}-\frac{1}{b_i})\frac{x_2G(x_1)-x_1G(x_2)}{x_2-x_1}$.

\end{lemma}
\begin{proof}
First, suppose that $\chi_i$ is optimal and suppose and $x>b_i$ and $x\neq x_1$, $x\neq x_2$. We will consider a new dyad $\hat{\chi_i}$ with support on $x_1$ and $x$ where:
\begin{eqnarray*}
  b_i &=& x_1 \eta + x (1-\eta)\\
\eta &=& \frac{x-b_i}{x-x_1} \\
H(\hat\chi_i) &=& \frac{(b_i-x_1)G(x) + (x-b_i)G(x_1)}{x - x_1}.
\end{eqnarray*}
Since $\chi_i$ is optimal, then $0 \le H(\chi_i) - H(\hat\chi_i)$ which implies that:
\begin{eqnarray*}
0 &\leq& \frac{(b_i-x_1)G(x_2) + (x_2-b_i)G(x_1)}{x_2 - x_1} - \frac{(b_i-x_1)G(x) + (x-b_i)G(x_1)}{x-x_1}. \end{eqnarray*}
Making use of the fact that $(x_2-x_1)$, $(x-x_1)$ and $(b_i-x_1)$ are all positive allows that:
\begin{eqnarray*}
0 &\leq& (x-x_1)(b_i-x_1)G(x_2) +(b_i-x_1)(x_2-x) G(x_1) - (x_2-x_1)(b_i-x_1)G(x) \\
0 &\leq& \frac{G(x_1)(x_2-x)+G(x_2)(x-x_1)}{(x_2-x_1)} -G(x) \\
&=& \frac{x(G(x_2)-G(x_1))+x_2G(x_1)-x_1G(x_2)}{x_2-x_1} -G(x) \\
&=& \frac{x H(\chi_i)}{b_i}+ (1-\frac{x}{b_i})\frac{x_2G(x_1)-x_1G(x_2)}{x_2-x_1} -G(x)
\end{eqnarray*}
and thus,
$$ \frac{G(x)}{x} \leq \frac{H(\chi_i)}{b_i}+ (\frac{1}{x}-\frac{1}{b_i})\frac{x_2G(x_1)-x_1G(x_2)}{x_2-x_1}. $$

Now consider the case where $x<b_i$, $x\neq x_1$, and $x\neq x_2$. We will consider a new dyad $\hat{\chi_i}$ with support on $x$ and $x_2$ where:
\[ H(\hat\chi_i) = \frac{(b_i-x)G(x_2) + (x_2-b_i)G(x)}{x_2 - x}\]
Since $\chi_i$ is optimal, then $0 \le H(\chi_i) - H(\hat\chi_i)$ which implies that:

\begin{eqnarray*}
0 &\leq& H(\chi_i) -H(\hat{\chi_i}) \\
0 &\leq& \frac{ (x_2-b_i)G(x_1)+(b_i-x_1)G(x_2)}{x_2-x_1} 
-\frac{ (x_2-b_i)G(x)+(b_i-x)G(x_2)}{x_2-x} \\
0 &\leq& G(x_1)(x_2-b_i)(x_2-x) +G(x_2)(x_2-b_i)(x-x_1) -G(x)(x_2-b_i)(x_2-x_1)\\
0 &\leq& G(x_1)(x_2-x) + G(x_2)(x-x_1) -G(x)(x_2-x_1)
\end{eqnarray*}
So we have,

$$ G(x) \le \frac{x(G(x_2)-G(x_1)) + x_2G(x_1) - x_1G(x_2)}{x_2-x_1} $$
and thus:
$$\frac{G(x)}{x} \le \frac{H(\chi_i)}{b_i} +(\frac{1}{x}-\frac{1}{b_i})\frac{x_2G(x_1)-x_1G(x_2)}{x_2-x_1} $$

For the converse, let us abbreviate the bound on $G(x)$ as $G(x)\le \frac{H(\chi)-A}{b_i}x+A$ where $A = \frac{x_2G(x_1)-x_1G(x_2)}{x_2-x_1}$. Consider any alternate dyad $\hat{\chi_i}$ supported on some $y_1$ and $y_2$, then,
\begin{eqnarray*}
H(\hat{\chi_i}) &=& \frac{ (b_i-y_1) G(y_2) + (y_2-b_i)G(y_1)}{y_2-y_1} \\
&\le& \frac{b_i-y_1}{y_2-y_1}\left[ \frac{H(\chi)-A}{b_i}y_2+A\right] + \frac{y_2-b_i}{y_2-y_1}\left[\frac{H(\chi)-A}{b_i}y_1+A\right] \\
&=& H(\chi)-A+A \\
&=& H(\chi).
\end{eqnarray*}
Thus $\chi$ is an optimal dyad.

\end{proof}

Assembling the conditions on dyads from lemma~\ref{lem:dyad_bound} together gives that each $f_i$ implies a linear constraint on $G(x)$, with a slope and intercept that can be calculated from any dyad in the support of $f_i$:

\begin{theorem}
$\{f_i\}$ is a Nash equilibrium if and only if for each $i$ there exists $A_i$ such that $G(x)\le \frac{H(f_i)-A_i}{b_i}x + A_i$ for all $x$ with equality if $x \in \sigma(f_i)$. Moreover, $$ A_i = \frac{x_2G(x_1)-x_1G(x_2)}{x_2-x_1},$$ for any dyad supported on $x_1,x_2\in \sigma(f_i)$. 
 
\label{thm:G_equality}
\end{theorem}
\begin{proof}
If $\{f_i\}$ is a Nash equilibrium then by Lemma~\ref{lem:dyad_bound} we have the relevant inequality on $G(x)$. To see that this is an equality for all $x\in\sigma(f_i)$ consider any ${x_1,x_2,x_3,x_4}\in \sigma(f_i)$. Let $\chi_0 = \chi(x|\theta_0)$,  $\chi_1 = \chi(x|\theta_1)$ and $\chi_2 =\chi(x|\theta_2)$ where $\theta_0 = \{x_1,x_2,b_i\}$, $\theta_1 = \{x_1,x_4,b_i\}$ and $\theta_2 = \{x|x_3,x_2,b_i \}$. 
Further, let $A_i = \frac{x_2G_(x_1)-x_1G(x_2)}{x_2-x_1}$ giving that $x_1$ and $x_2$ satisfy $G(x) = \frac{H(\chi)-A_i}{b_i}x + A_i$. Since $f_i$ is optimal, then $H(f_i)=H(\chi_0)=H(\chi_1) = H(\chi_2)$, and thus the bound on $G$ is satisfied at $x_1$ and $x_2$. We will now show that $x_3$ and $x_4$ also satisfy this bound and agree on the value of $A_i$. 

Similar to the algebra in the proof of Lemma~\ref{lem:dyad_bound}, $H(\chi_0)-H(\chi_1)=0$ and $H(\chi_0)-H(\chi_2)=0$ imply:
\begin{eqnarray*}
\frac{G(x_3)}{x_3} &=& \frac{H(\chi_0)}{b_i} + A_i(\frac{1}{x_3}-\frac{1}{b_i}) \\
\frac{G(x_4)}{x_4} &=& \frac{H(\chi_0)}{b_i} + A_i(\frac{1}{x_4}-\frac{1}{b_i}).
\end{eqnarray*}
Subtracting these yields:
\begin{equation*}
\frac{x_4G(x_3)-x_3G(x_4)}{x_4-x_3} = A_i.    
\end{equation*}
Thus, for each $x\in \sigma(f_i)$, $\frac{G(x)}{x}=\frac{H(f_i)}{b_i} + A_i(\frac{1}{x}-\frac{1}{b_i}) $.

For the opposite direction, consider any $i$ and dyad $\chi = \chi(x|\theta)$ for $\theta = \{x_1,x_2,b_i\} $, and assume that there exists $H_i$ and $A_i$ such that $G(x)<\frac{H_i-A_i}{b}x -A_i$ with equality for all $x\in\sigma(f_i)$, then:
\begin{eqnarray*}
H(\chi) &=& \frac{(x_2-b_i)G(x_1)+(b_i-x_1)G(x_2)}{x_2-x_1} \\
&\le& \frac{1}{x_2-x_1}[A_i(x_2-x_1)+b_i(x_2-x_1)\frac{H_i-A_i}{b_i}] \\
&=& H_i,
\end{eqnarray*}
with equality when both $x_1$ and $x_2$ are in the support of $f_i$. Thus, $f_i$ is optimal for all $i$ and $\{f_i\}$ is a Nash Equilibrium where $H_i = H(f_i)$ and the form of $A_i$ follows from the previous direction.

\end{proof}

Interestingly, the above theorem holds true for the Multiplayer Game so long as $G(x)$ is replaced by $\sum_{j\neq i} f_j $. 

That $G(x)$ is bounded by a set of lines immediately gives that $G$ is concave for all Nash Equilibrium:

\begin{lemma}
\label{lem:Nash_concave}
If strategies $\{f_i\}$ are a Nash equilibrium then $G(x)$ is concave and $G(0)=0$. 
\end{lemma}
\begin{proof}
By lemma \ref{lem:dyad_bound}, $G(x)$ is less than or equal to an intersection of half-spaces, one for each $f_i$ each with positive slopes. By corollary \ref{thm:G_equality}, $G(x)$ equals these bounds for all $x\in \sigma(f_i)$ and for each $i$. Since $G'(x)>0$ only when some $f_i$ is supported, then $G(x)$ must be a concave function. To see that $G(0)=0$, consider the counterfactual where $G(0)>0$ because for some $i$ $f_i$ includes a dyad $\{0,x_1\}$. Notice that for $i$, the difference for playing dyad $\{\epsilon,x_1\}$ instead of $\{0,x_1\}$ is at least $\frac{x_1-b_i}{b_i-\epsilon}(A_i-\frac{\epsilon}{x_1}) $. As $\{0,x_1\}$ was optimal, then this difference must be non-positive for all $\epsilon<b_i$, implying that $A_i=0$. Since $A_i = \frac{x_1G(0) - 0\cdot G(x_1)}{x_1-0}$ then  $G(0)=0$.  \end{proof}

Before we arrive at our main theorem, consider the following technical corollaries.  Let $\Bar{\sigma} (f_i) = [\sigma_i^-,\sigma_i^+]$ where $\sigma_i^- = inf(\sigma(f_i))$ and $\sigma_i^+ = sup(\sigma(f_i))$.
\begin{corollary}
\label{cor:sigma_bar_fill}
In a Nash equilibrium, $G(x) = \frac{H(f_i)-A_i}{b_i}x +A_i$ for all $x$ in  $\bar{\sigma}(f_i)$.
\end{corollary}
\begin{proof}
Since by theorem \ref{thm:G_equality} $G(x) = \frac{H(f_i)-A_i}{b_i}x +A_i$ for all $x$ in  $\sigma(f_i)$ the result follows immediately from the concavity of $G(x)$. 
\end{proof}

\begin{corollary}
\label{cor:overlap}
In a Nash equilibrium, if $\sigma_i^+ > \sigma_j^-$ and $b_i<b_j$ then $G(x) = \frac{H(f_i)-A_i}{b_i}x +A_i$ for all $x$ in  $\bar{\sigma}(f_i)\cup \bar{\sigma}(f_j)$.
\end{corollary}
\begin{proof}
By corollary \ref{cor:sigma_bar_fill}, $G(x) = \frac{H(f_i)-A_i}{b_i}x +A_i = \frac{H(f_j)-A_j}{b_j}x +A_j$ for all $x\in[\sigma_j^-, \sigma_i^+]$ and since $b_i<b_j$ this interval is nonempty. Thus $A_i=A_j$ and $\frac{H(f_i)-A_i}{b_i}= \frac{H(f_j)-A_j}{b_j}$.  
\end{proof}

\begin{theorem}
$\{f_i\}$ is a Nash equilibrium if and only if 
\begin{enumerate}[I.]
    \item $g(x)$ is non-increasing
    \item $g(x)$ is constant on $\bar{\sigma}(f_i) \forall i$
\end{enumerate}
\label{thm:NashIFF}
\end{theorem}
\begin{proof}
First, notice that given $I\!I$:
\begin{eqnarray*}
	H(f_i) &=& \int_0^\infty f(x)G(x)dx \\
	&=& \int_{\sigma^-}^{\sigma^+} f(x)G(x)dx \\
	&=& \int_{\sigma^-}^{\sigma^+} f(x)\big(G(\sigma^-)+g(b_i)(x-\sigma^-)\big)dx \\
	&=& \int_{\sigma^-}^{\sigma^+} f(x) x g(b_i) dx + \int_{\sigma^-}^{\sigma^+} f(x) \big(G(\sigma^-)-g(b_i) \sigma^-\big) dx \\
	&=& b_i g(b_i) + G(\sigma^-) - \sigma^- g(b_i) = G(\sigma^-) + g(b_i) (b_i - \sigma^-) \\
	H(f_i) &=& G(b_i)
\end{eqnarray*}

Then, for some $x_1 < b_i < x_2$, $x_1,x_2 \in \bar{\sigma}(f_i)$
\begin{eqnarray*}
G(b_i)-b_i g(b_i) &=& G(b_i) + b_i (\frac{-x_2+x_1}{x_2-x_1})g(b_i) \\
	&=& G(b_i) + \frac{x_2 x_1 - x_2 b_i - x_1 x_2 + x_1 b_i}{x_2-x_1}g(b_i) \\
	&=& \frac{x_2 (G(b_i) + (x_1-b_i) g(b_i)) - x_1 (G(b_i)+(x_2-b_i) g(b_i))}{x_2-x_1} \\
	&=& \frac{x_2 G(x_1) - x_1 G(x_2)}{x_2-x_1} = A_i
\end{eqnarray*}

Thus, $I\!I$ implies:
\begin{eqnarray*}
	G(x) &=& G(b_i) + (x-b_i) g(b_i) \\
	&=& \frac{G(b_i)-G(b_i)+b_i g(b_i)}{b_i} x + G(b_i)-b_i g(b_i) \\
	&=& \frac{H(f_i)-A_i}{b_i} x+A_i \quad \forall x \in \bar{\sigma}(f_i)
\end{eqnarray*}

Since $I$ implies $G(x)$ is concave, then $G(x)$ is bounded everywhere by it's tangent lines, and thus: $$G(x) \leq \frac{H(f_i)-A_i}{b_i} x+A_i \  \forall x \notin \bar{\sigma}(f_i) \forall i$$

Therefore, conditions $I$ and $I\!I$ together conform to Theorem \ref{thm:G_equality} and thus $\{f_i\}$ is a Nash equilibrium.

The converse, that a Nash equilibrium has the properties of $I$ and $I\!I$, are shown by Lemma \ref{lem:Nash_concave} and Corollary \ref{cor:sigma_bar_fill} respectively.
\end{proof}

Theorem~\ref{thm:NashIFF} provides a concrete and simple test to verify whether a population is in a Nash equilibrium and will serve as the backbone of Algorithm~\ref{alg:FindNash} in the next section. Next we explore the constraints on the competitive structure that being in a Nash equilibrium imposes.

\begin{lemma}
If $g(b_i)<g(b_j)$ then in a pairwise competition $i$ beats $j$ almost surely. \label{lem:gprimewins}
\end{lemma}
\begin{proof}
From theorem \ref{thm:G_equality}, $g(x)=\frac{H(f_i)-A_i}{b_i}$ over the support of each $f_i$. Then since $G(x)$ is concave by lemma \ref{lem:Nash_concave}, $g(b_i)<g(b_j)$ only when the support of $f_i$ is greater than the support of $f_j$ (i.e. $\sigma_j^+\le \sigma_i^-$. Thus in a pairwise competition, $i$ beats $j$).  
\end{proof}

Combining these statements together, and observing that the bound $G(x) = \frac{H(f_i)-A_i}{b_i}x +A_i$ implies that $G(b_i)=H(f_i)$ gives the following theorem.

\begin{thm}
In a Nash equilibrium, $H(f_i) = G(b_i)$ and if $g(x_1)> g(x_2) $ then all players with budgets less than or equal to $x_1$ always lose to players with budgets greater than or equal to $x_2$.
\end{thm}

An implication of this theorem is that $g$ has a significant impact on the competitive landscape of competition, and motivates the following definition.

\begin{definition}[League]
A league $L_{c}$ consists of all players $i$ such that $g(b_i)=c$.
\end{definition}

Since $G$ is concave by \ref{lem:Nash_concave}, then $L_c$ contains all players $i$ with $b_i\in(b^-,b^+)$ for some $b^-$ and $b^+$. By lemma \ref{lem:gprimewins}, in a Nash equilibrium, players in $L_c$ always beat players in $L_d$ if and only if $c<d$. While, competitions between leagues are determined entirely by the leagues (and thus the budgets) competitions inside a league are not deterministic, and may involve players with lower budgets beating those with higher budgets and non-transitive cycles as illustrated in Figure~\ref{fig:NashCycles}.  

\begin{figure}[t!]
    \centering
    \includegraphics[width=.98\textwidth]{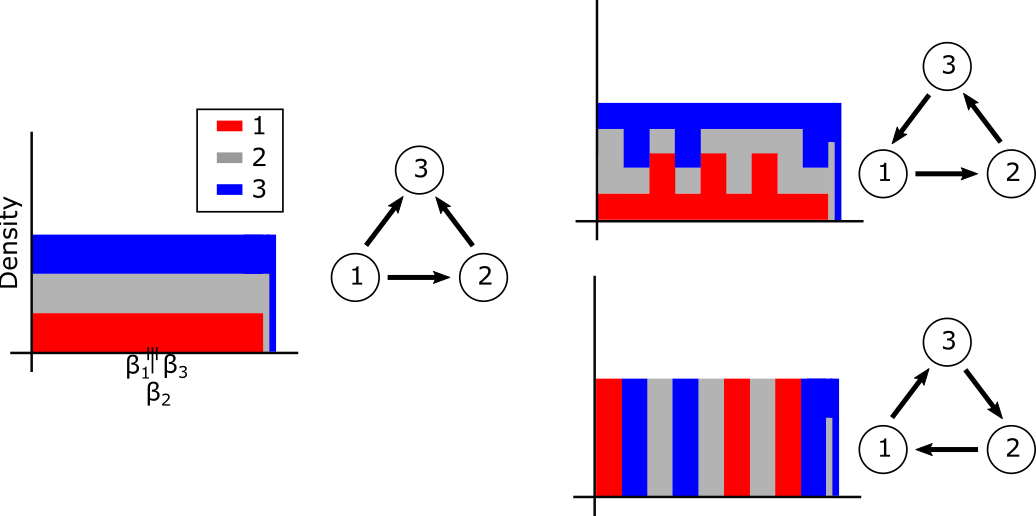}
    \caption{In a Nash equilibrium, the distributions inside a league are not uniquely determined. In a sub-population consistent Nash equilibrium the distributions are transitive according to budgets (left), but alterations to those distributions may be able to allow different outcomes in a Nash Equilibrium, including nontransitive cycles (right). In the directed network, arrows represent endorsements towards the expected winner. }
    \label{fig:NashCycles}
\end{figure}

Nonetheless, there are Nash equilibrium where expected competition outcomes are concordant with budgets.

\begin{definition}[Sub-population Consistent Nash Equilibrium]
\label{def:SubConNE}
A population is in a sub-population consistent Nash equilibrium if for every $\beta>0$ the subset of players with budgets between $[0,\beta]$ and between $[0,\beta)$ are both Nash equilibrium as well. 
\end{definition}

In the next section, we present Algorithm~\ref{alg:FindNash} which finds the unique sub-population consistent Nash equilibrium for any discretized budget population. As this construction will make clear, for discretized budget populations, sub-population consistent Nash equilibria are transitive in budget, such that if $b_i>b_j$ then $H(f_i,f_j)>\frac{1}{2}$.

To see that this construction ensures transitivity on the budgets notice that if $b_i>b_j$ then either $i$ and $j$ are in separate leagues, in which case $H(f_i,f_j)=1$ or $j$'s support lies in the support of $i$'s, and $f_i$ is constant over $i$. Since $i$ plays a constant weight over the support of $j$ then, as in the 2-player game, the outcome of $i$ and $j$ is the same as if $j$ played a point mass at $b_j$. Since $b_i>b_j$ and the distributions produced by algorithm~\ref{alg:FindNash} skew left, then the median of $f_i$ lies at or to the right of $b_i$, and thus $H(f_i,f_j)>\frac{1}{2}.$

\section{Nash Equilibria on Discretized Populations }
To investigate the existence of a Nash equilibrium consider the following discretization of a population budget distribution. Suppose, that an infinite population can be broken into $n$ distinct groups $\{1, \ldots, n \}$, so that each member of group $i$ has budget $\beta_i$, with $\beta_i>\beta_j$ if $i>j$ (so that budget distribution $B(y)$ is composed of $n$ atoms). Further, suppose that the relative sizes of these groups are given by $k_i$ so that $\sum_i k_i = 1$. In this scenario one can calculate a sub-population consistent Nash equilibrium for populations with discretely supported budgets as in Algorithm~\ref{alg:FindNash}.

\begin{algorithm}
\DontPrintSemicolon
\KwData{$\vec{\beta}$ and $\vec{k}$, both vectors of length $n$, where $\sum_i^n(k_i) = 1$.}
\KwResult{piece-wise constant $g$ expressed as partitions $P$, with associated constant values $Y$}
$ P \longleftarrow \{0\}$ \\
$ Y \longleftarrow \{\infty\}$ \\
\For{$i \in [1,\ldots,n]$}{
fill($P$, $Y$, $\beta_i$, $k_i$) \# as described in Alg. \ref{alg:Level}
}
\caption{find Sub-population Resource Allocations\label{alg:FindNash}}
\end{algorithm}

\begin{algorithm}
\DontPrintSemicolon
\KwData{$P$, $Y$, $\beta$, $k$.}
\KwResult{Updates partitions $P$, and values $Y$ to include probability mass $k$ with average budget $\beta$}
\If{$\beta > P[\mathtt{end}]$}{
$p \longleftarrow \beta + (\beta-P[\mathtt{end}])$ \\
$y \longleftarrow \frac{k}{p-P[\mathtt{end}]}$ \\
\If{$y < Y[\mathtt{end}]$}{
append $p$ to $P$ \\
append $y$ to $Y$ \\
$\mathbf{Return}$
}
}
$a \longleftarrow k + (P[\mathtt{end}]-P[\mathtt{end}-1]) \cdot Y[\mathtt{end}]$ \\
$b \longleftarrow -2 \cdot k \cdot \beta - (P[\mathtt{end}]^2-P[\mathtt{end}-1]^2) \cdot Y[\mathtt{end}]$ \\
$c \longleftarrow -P[\mathtt{end}-1]^2 \cdot (k+(P[\mathtt{end}]-P[\mathtt{end}-1]) \cdot Y[\mathtt{end}]) + P[\mathtt{end}-1] \cdot (2 \cdot k \beta + (P[\mathtt{end}]^2-P[\mathtt{end}-1]^2) \cdot Y[\mathtt{end}])$\\
$p \longleftarrow \frac{-b + \sqrt{b^2 - 4 \cdot a \cdot c}}{2\cdot a}$ \\
$y \longleftarrow \frac{k + (P[\mathtt{end}]-P[\mathtt{end}-1]) \cdot Y[\mathtt{end}]}{p-P[\mathtt{\mathtt{end}-1}]}$ \\
\uIf{$y < Y[\mathtt{\mathtt{end}-1}]$}{
$P[\mathtt{end}] \longleftarrow p$\\
$Y[\mathtt{end}] \longleftarrow y$\\
}
\Else{
\# the mass has filled its area, and overflows onto the previous area\\
$\mathtt{used} \longleftarrow (Y[\mathtt{end}-1]-Y[\mathtt{end}])*(P[\mathtt{end}]-P[\mathtt{end}-1])+Y[\mathtt{end}-1]*(p-P[\mathtt{end}])$\\
$\mathtt{leftover} \longleftarrow k - \mathtt{used}$\\
$Y[\mathtt{end}] \longleftarrow Y[\mathtt{end}-1]$\\
$P[\mathtt{end}] \longleftarrow p$\\
$DEL(Y[\mathtt{end}-1])$\\
$DEL(P[\mathtt{end}-1])$\\
$fill(P,Y,\beta,\mathtt{leftover})$
}
\caption{fill\label{alg:Level}}
\end{algorithm}

Although the implementation of Algorithm~\ref{alg:FindNash} is relatively opaque, the operating principle can be readily understood. Algorithm~\ref{alg:FindNash} will construct $g$ by sequentially using Algorithm~\ref{alg:Level} to add distributions with mass $k_i$ and increasing budgets $\beta_i$. In order to guarantee that the associated distribution is a Nash equilibrium, we ensure that $g$ is non-increasing and constant on the support of any previous agent's distributions after the addition of each new distribution. In order to ensure these properties consider the following analogy for the operation of Algorithm~\ref{alg:Level}. The current distribution $g$ is a terraced structure and the new distribution of agent $i$ is a fixed amount of $k_i$ liquid concrete held between that structure and an outer retaining wall initially located far in the distance, as in Figure~\ref{fig:water}. When the retaining wall is distant the mean of the concrete is much larger than $\beta_i$, but as that retaining wall is brought closer to the structure the mean of the concrete continuously decreases, while it's height increases. As the height of the concrete increases it may begin to overflow various terraces of the structure, but each time it does so it only coarsens the existing terraces. When the retaining wall has been brought close enough so that the concrete has the appropriate mean, the concrete is solidified, adding to the structure for the next additional distribution. Since the mean monotonically decreases as the retaining wall is brought in, the resulting constant distribution is unique under our assumptions.

Since this construction produces $g$ that are constant on the support of agents the only important point to check is that the mean of the new distribution is achieved before the retaining wall hits the structure. That this condition on the retaining wall is met relies on two facts: first, distributions produced by our algorithm are either uniform or skew left, which implies that each $j$, $\beta_{j}\ge\frac{1}{2}(\sigma^-_{j} + \sigma^+_{j})$; and second, if $f_i$ overlaps $f_{i-1}$ then $f_i$ overlaps all of $f_{i-1}$. Together these give that at the point where the retaining wall hits the structure the mean of $\beta_i$ must be less than or equal to $\beta_{i-1}$, contradicting our ordering on $\beta$'s.      

Rather than animate the continuous approach of a retaining wall, Algorithm~\ref{alg:Level} implements this logic by iteratively calculating the retaining wall's location, each time assuming it may overflow onto one additional terrace (calculations in the appendix). This iteration terminates once the resulting distribution does not overflow the next terrace.   

\begin{figure}[hbt!]
    \centering
    \includegraphics[width=.9\textwidth]{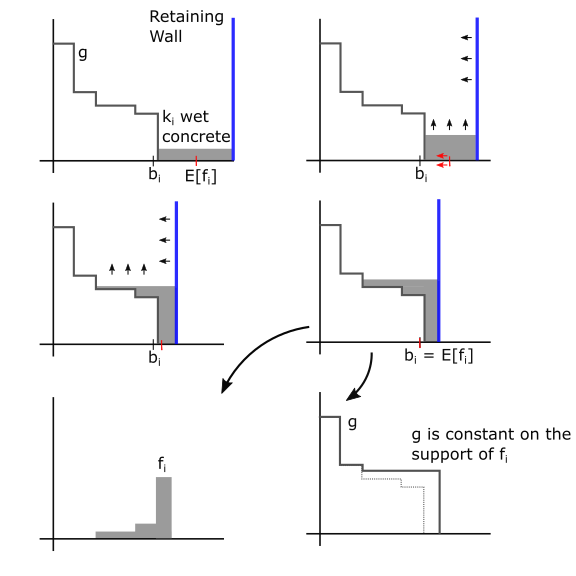}
    \caption{The operating principle of Algorithm~\ref{alg:Level} can be conceptualized as the movement of $k_i$ wet concrete placed between a retaining wall and the current $g$. As the retaining wall is brought in the mean location of the wet concrete decreases, while the height of the wet concrete increases, possibly overflowing onto the terraces of $g$. When the mean location of the concrete equals the target budget the concrete is allowed to set, creating the distribution for $f_i$ and an updated $g$.  }
    \label{fig:water}
\end{figure}

\section{Sub-Leagues and Non-Transitivity}

Thus far we have shown that for any discretized budget distribution there always exists a sub-population consistent Nash equilibrium, and that inside any Nash equilibrium with different leagues, competitions between leagues are almost surely determined by budgets. In this section we will further explore several notions of transitivity and propose that the structures inside a Nash equilibrium are better understood using the concept of sub-leagues.

It is worth distinguishing between outcomes which are likely in expectation and those which have probability one. As in Figure~\ref{fig:algorithm_examples} (right) the competitive relationships can be visualized as a directed network with almost surely (probability one) outcomes shown as red arrows, and expected outcomes shown as black arrows. These directed networks display some forms of transitivity but not others. One established notion of transitivity is that of weak stochastic transitivity:

\begin{definition}[Weak Stochastic Transitivity]
A set of strategies $\{f_i\}$ has weak stochastic transitivity if for any $i$, $j$ and $k$, such that $H(f_j,f_i)\ge 0.5$ and $H(f_k,f_j) \ge 0.5$ implies that $H(f_k,f_i)\ge 0.5$.
\end{definition}
  
As we have shown, sub-population consistent Nash equilibria always have weak stochastic transitivity because such Nash equilibria are transitive on their budgets, but, as in Figure~\ref{fig:NashCycles}, not all Nash equilibria have weak stochastic transitivity. The accompanying established concept of strong stochastic transitivity posits that:

\begin{definition}[Strong Stochastic Transitivity]
A set of strategies $\{f_i\}$ has strong stochastic transitivity if for any $i$, $j$ and $k$, such that $H(f_j,f_i)\ge 0.5$ and $H(f_k,f_j) \ge 0.5$ implies that $H(f_k,f_i)\ge \max(H(f_j,f_i),H(f_k,f_j))$.
\end{definition}

Nash equilibrium and even sub-population consistent Nash equilibrium need not necessarily have strong stochastic transitivity as we will show for a special case of these conditions.

Another notion of transitivity considers whether almost surely outcomes are transitive:

\begin{definition}[Transitive in Certainty]
A set of strategies $\{f_i\}$ is transitive in certainty if for any $i$, $j$ and $k$, such that $H(f_j,f_i)=1$ and $H(f_k,f_j) = 1$ implies that $H(f_k,f_i)= 1$.
\end{definition}  

In fact, the Population Lotto game is always transitive in certainty, regardless of whether the population is in a Nash equilibrium or not. Namely, $H(f_j,f_i)=1$ if and only if the support of $f_j$ is to the right of the support $f_i$, and $H(f_k,f_j) = 1$ if and only if the support of $f_k$ is to the right of $f_j$, which together imply that the support of $f_k$ is to the right of $f_i$ and thus $H(f_k,f_i) = 1$. Given that the Population Lotto game is always transitive in certainty, and sub-population consistent Nash equilibria have weak stochastic transitivity, one might expect that almost surely outcomes are transitive with outcomes in expectation in the following ways:  
\begin{definition}[Transitive in Dominance]
A set of strategies $\{f_i\}$ is transitive in dominance if for any $i$, $j$ and $k$, such that $H(f_j,f_i) \ge 0.5$ and $H(f_k,f_j)=1$ implies that $H(f_k,f_i)= 1$.
\end{definition}  

\begin{definition}[Transitive in Establishment]
A set of strategies $\{f_i\}$ is transitive in establishment if for any $i$, $j$ and $k$, such that $H(f_j,f_i)=1$ and $H(f_k,f_j) \ge 0.5$ implies that $H(f_k,f_i)= 1$.
\end{definition}  
  
\begin{figure}[hbt!]
    \centering
    \includegraphics[width=.49\textwidth]{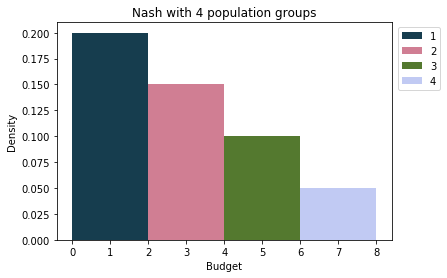}
    \includegraphics[width=.3\textwidth]{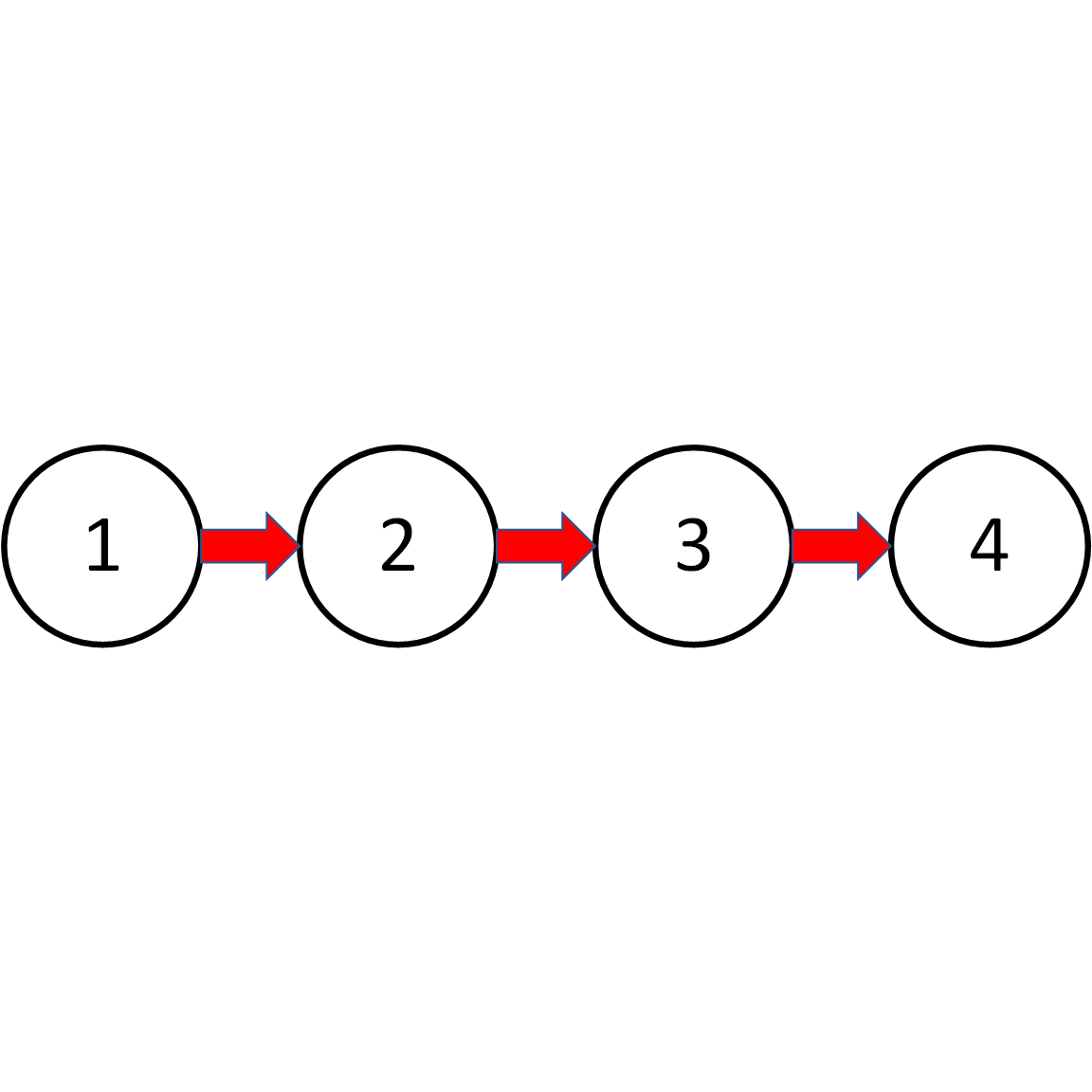}
    
    \includegraphics[width=.49\textwidth]{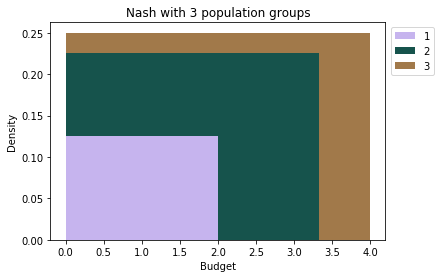}
    \includegraphics[width=.3\textwidth]{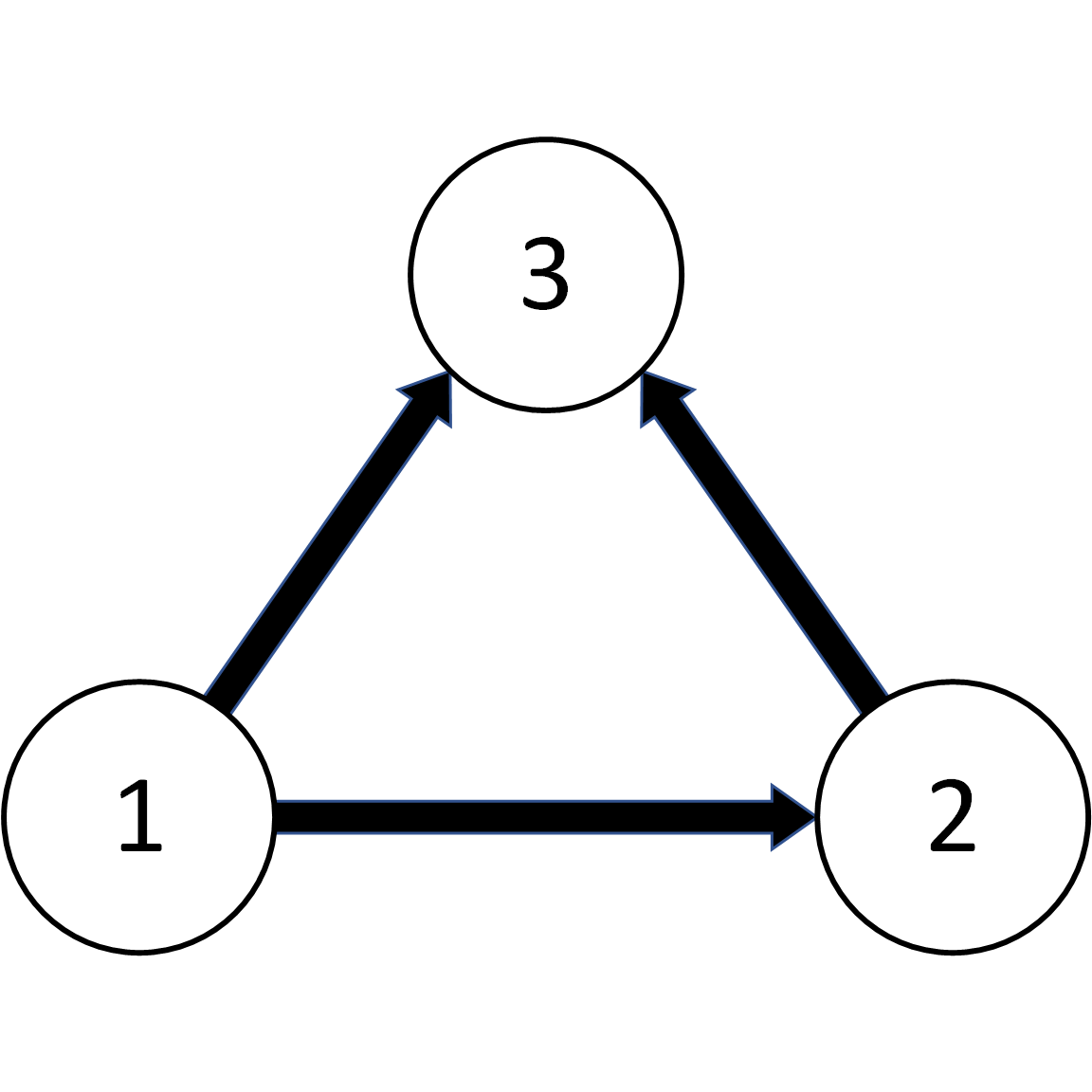}
    
    \includegraphics[width=.49\textwidth]{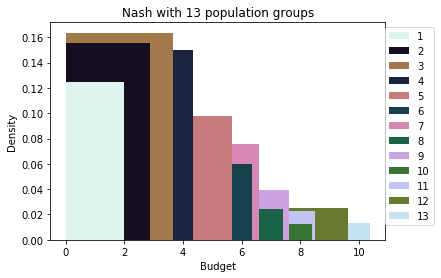}
    \includegraphics[width=.3\textwidth]{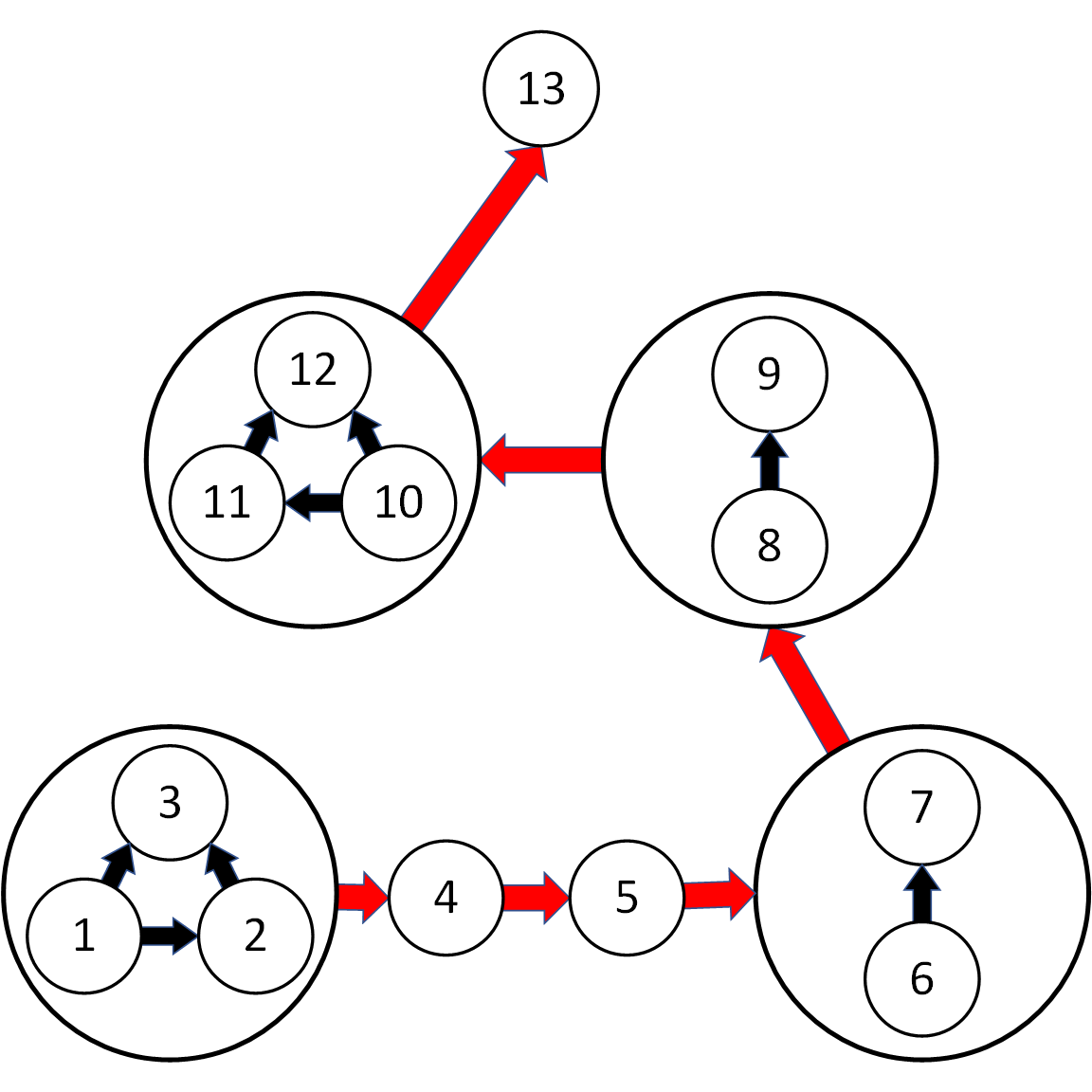}
    
    \includegraphics[width=.49\textwidth]{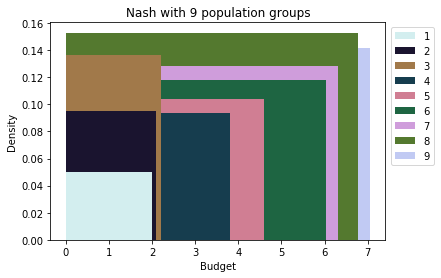}
    \includegraphics[width=.3\textwidth]{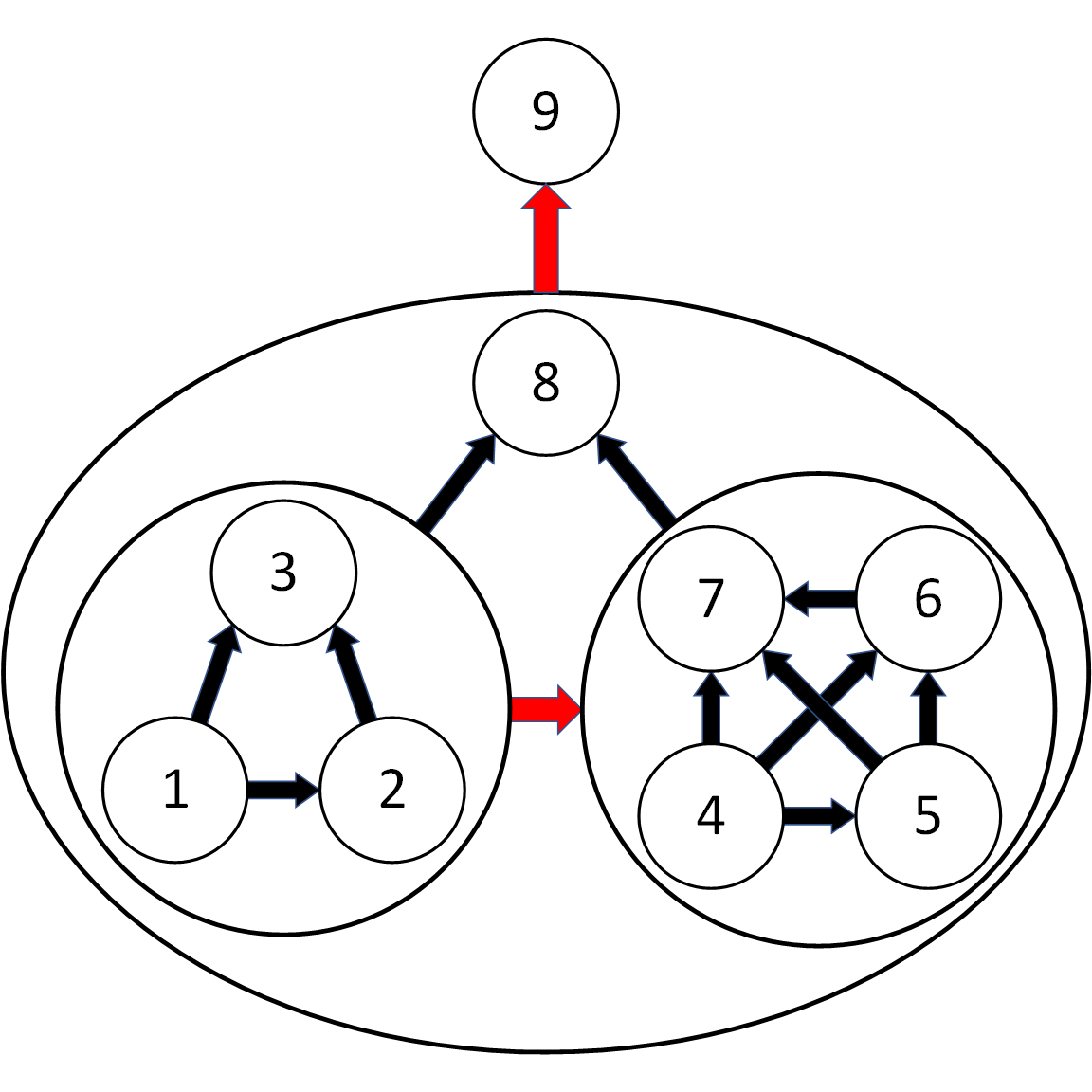}
    
    \caption{Examples sub-population consistent Nash equilibria represented as distributions (left) with competitive outcomes displayed as directed graphs (right). In the directed graphs, arrows represent endorsements where black arrows represent the direction of expected loss and red arrows represent guaranteed losses due to nodes being in different leagues or sub-leagues. Note that the bottom is not transitive in establishment as $H(f_4,f_3)=1$, $H(f_8,f_4)\ge 0.5$ but $H(f_8,f_3)<1$. }
    \label{fig:algorithm_examples}
\end{figure}

Note that both transitivity in dominance and establishment are two special cases of strong stochastic transitivity. Similarly note that a system can be weakly stochastic transitive and/or transitive in certainty without being transitive in dominance or establishment. Sub-population consistent Nash equilibria are always transitive in dominance, because in a sub-population Nash equilibrium $H(f_j,f_i)\ge 0.5$ implies that $\sigma_j^+\ge\sigma_i^+$ and $H(f_k,f_j)=1$ implies that $\sigma_k^-\ge\sigma_j^+\ge\sigma_i^+$ giving that $H(f_k,f_i)=1$. In contrast, as in Figure~\ref{fig:algorithm_examples} (bottom), sub-population Nash equilibria are not necessarily transitive in establishment, and thus also do not have strong stochastic transitivity. 

These results suggest Nash equilibria in the Population Lotto game, along with the more restrictive sub-population consistent Nash equilibria place some, but relatively loose constraints on potential non-transitivity. In contrast, the structure of sub-population consistent Nash equilibria can be further understood from a different perspective. 

To illustrate, consider the sub-population consistent Nash equilibrium for different budget distributions in Figure~\ref{fig:algorithm_examples}, which reveals several layers of structure. As in Figure~\ref{fig:algorithm_examples} (top) some distributions have sub-population consistent Nash equilibrium that cleanly separate different sub-populations into separate leagues with almost surely outcomes, while others may place all sub-populations into a single league (2nd from top). Moreover, sub-populations can exist in the same league even when there are many different leagues, as in Figure~\ref{fig:algorithm_examples} (2nd from bottom), and appear to create nested league like structure (bottom) which motivate the following definition:

\begin{definition}[Sub-league]
In a sub-population consistent Nash equilibrium, a set of sub-populations of a league are sub-leagues if they are not leagues but there exists a $\beta$ such that those sub-populations are separate leagues in a Nash equilibrium composed of only those with budgets less than or equal to $\beta$.  
\end{definition}

Examples of sub-leagues are sub-populations $\{1,2,3\}$ along with $\{4,5,6,7\}$ in Figure~\ref{fig:algorithm_examples} (bottom), because if sub-populations $8$ and $9$ were absent these two sets of populations would create separate leagues. Note that a league which contains multiple sub-leagues necessarily implies an interesting violation of transitivity in establishment so long as a sub-population shares support with any two disjoint sub-leagues.

Any budget distributions whose sub-population consistent Nash equilibrium have sub-leagues allow for Nash equilibria without sub-leagues, because it is possible to perform the intra-league modifications (ex. those in Figure~\ref{fig:NashCycles}), that destroy the sub-league structure while preserving the overall shape of $g$. Indeed, for there to be a sub-league there must be some distribution that not in the sub-league that overlaps at least one member of the sub-league, and this allows for budget neutral swaps that leave the shape of $g$ unchanged, but eliminate the sub-leagues and the sub-population consistent property. Thus, sub-leagues are not uniquely determined by budget distributions. In contrast, we believe that league structure is uniquely determined by the budget distribution:

\begin{conjecture}[League Structure Uniqueness]
For any budget distribution $\mathcal{B}$ all Nash equilibrium share the same $g$ and thus the same league structure.  
\end{conjecture}

Note that since theorem~\ref{thm:G_equality} states that $H(f_i)=G(b_i)$
then the League Structure Uniqueness conjecture would also imply that the Nash equilibrium expected payoffs for players are also uniquely determined by the budget distribution. 

Returning to the example equilibrium in Figure~\ref{fig:algorithm_examples} (bottom) with two disjoint sub-leagues, suggests an interesting phenomenon. Namely, absent sub-populations $8$ and $9$, sub-populations $1-7$ might imagine that there is a clear established hierarchical difference between their two sub-leagues, whereas the presence of $8$ places both sub-leagues into a single leagues and thus implies the existence of numerous other Nash equilibria with different competitive outcomes between sub-populations $1-7$. In this way, the addition of higher budget sub-populations to an established hierarchy can destabilize existing leagues by relegating previous leagues into sub-leagues, an effect we call the `bigger pond' effect, wherein moving smaller fish into a bigger pond with bigger fish diminishes the previously existing hierarchy of the smaller fish. Without the adjudication of particular potentially contentious historical analogs, we feel that this effect can serve as a useful conceptual analogy.

\section{Conclusion and Future Work}
Our introduction and analysis of the population Lotto game has provided a useful criterion to determine whether a population is in a Nash equilibrium and established useful characterizations of all Nash equilibrium. Algorithm~\ref{alg:FindNash} establishes that any discretized budget distribution has a sub-population consistent Nash equilibrium, and that equilibrium can be understood as dividing the population into different leagues and sub-leagues. Thus, strategic budget allocation can constrain the structure of competitive outcomes in a population where the league structure of a Nash equilibrium determines the almost surely competitions and corresponds to a rigid hierarchy of outcomes. If, as we conjecture, the league structure is uniquely determined by the budget distribution, then strategic resource allocation in the population Lotto game has strong and clear implications on establishing structure in outcomes. On the other hand, since structure inside leagues is not uniquely determined, any existing prevalence of sub-league structure would require additional motivation. For instance, if players have uncertainty about the sizes of richer sub-populations, then players are more likely to end up in sub-population consistent Nash equilibria, as those equilibria are robust to either the deletion or addition of a sub-population that increases the maximum budget. 

We believe this work motivates a number of interesting future directions. Our results on the Population Lotto games provide significant intuition on the multiplayer Lotto game. For instance, for populations where for all $i$, $g-f_i$ retains the same overall league structure, we believe that much of the analysis of the population Lotto game should apply. 

Another consideration is when the structure of potential match-ups is not uniform. Indeed, while professional sports teams will almost surely win against non-professional college and high school teams, obeying a league structure, such a match-up does not, as a rule, regularly occur. Instead, if players or sub-populations inhabited a network, with match-ups only occurring across the edges of the network then the Nash equilibria, if they even exist, is likely quite different. Indeed, under reasonable assumptions a network of Lotto players might create non-stationary and interesting dynamics, complicating attempts to explain the transitivity of human competition data by strategic resource allocation.    

While algorithm~\ref{alg:FindNash} deals explicitly with discretized populations, we believe that similar logic holds with only a few modifications to continuous budget distributions. Whereas the output of algorithm~\ref{alg:FindNash} is the strategy of the next group in order of increasing budgets, the corresponding distribution for a continuous budget distributions would represent the partial derivative of aggregated strategies against sub-population budget. In that continuous setting, we believe the natural generalization would be distributions which are uniform distributions with a delta spike at their right endpoint. Indeed, the terraced output $f_i$ of algorithm~\ref{alg:FindNash} (see the bottom left of Figure~\ref{fig:water}) can be interpreted as an integration over such uniform plus delta distributions where the tallest segment on the right side of $f_i$ is composed of delta spikes and the other terrace-like segments of $f_i$ are cut horizontally to construct uniform distributions. The exact formulation of these observations, and the conditions on $\mathcal{B}$ such that all agents play only a delta spike on their budget we leave for future work. 

Lastly, a natural extension of this work is applying the results to random pair-wise population versions of the Colonel Lotto and Blotto games. Indeed, it is frequently the case that the solution to the continuous relaxations of the discrete games yields insights back to the original discrete versions.

\section*{Acknowledgments}
This work was partly supported by the Western Alliance to Expand Student Opportunities (WAESO) Scholarships for Faculty-Directed Undergraduate Research Projects. The authors also acknowledge productive conversations and prior work on related problems with Max Mercer and Hala King.

\section*{Appendix}
\subsection*{Calculations for Algorithm~\ref{alg:Level}}
Algorithm~\ref{alg:Level} repeatedly calculates what the height and location of a new distribution for $f_i$ should be. Inside the central loop of Algorithm~\ref{alg:Level} supposes that $f_i$ has overflowed on top of exactly one terrace, and calculates the height, and right endpoint so that a distribution satisfies the budget and mass constraints. In particular, suppose the right most terrace of $g$ has height $y_1$ and is supported on $[x_2,x_1]$. Algorithm~\ref{alg:Level} supposes that $f_i$ is the union of a rectangular function of height $y-y_1$ supported on $[x_2,x_1]$ and a rectangular function of height $y$ supported on $[x_1,x]$, satisfying the mass and budget constraints:
\begin{eqnarray*}
m &=& (x_{1}-x_{2})(y-y_{1})+(x-x_{1})y \\
b &=& \tfrac{1}{2m}[(x_{1}^2 - x_{2}^2)(y-y_{1})+(x^2 - x_{1}^2)y   ].
\end{eqnarray*}
Solving the first equation gives that $y = \frac{m+y_{1}(x_{1}-x_{2})}{x-x_{2}}$. The second equation is thus $(m+(x_1-x_2)y_1)x^2 -(2mb+(x_1^2-x_2^2)y_1)x - x_2^2(m+ (x_1-x_2)y_1) + x_2(2mb + (x_1^2-x_2^2)y_1)=0$ which can be solved for $x$. 

If $y$ is smaller than the height of the next terrace, $y_2$ then the algorithm terminates. Instead, if $y$ is larger than the height of the next terrace then it is clear that the distribution should overflow the next terrace, in which case mass $(x_1-x_2)(y_2-y_1)$ is used from $i$ to fill in the $[x_2,x_1]$ terrace and $i's$ budget is adjusted so that the algorithm can attempt to fill in another single terrace.

\bibliographystyle{plain}
\bibliography{blotto}

\end{document}